\documentclass[11pt]{article}

\usepackage{etex}
\usepackage{graphics}
\usepackage{color}
\usepackage{cancel}
\usepackage{stmaryrd}
\usepackage{amsmath}
\usepackage{amsthm}
\usepackage{amssymb}
\usepackage{proof}
\usepackage{verbatim,cmll}
\usepackage{setspace}
\usepackage{lscape}
\usepackage{latexsym,relsize}
\usepackage{amscd}
\usepackage{multicol}
\usepackage{bussproofs}
\usepackage[position=top]{subfig}
\usepackage{leqno}
\usepackage{tikz}
\usepackage{authblk}
\usepackage{multicol}
\usepackage{tkz-graph}
\usetikzlibrary{arrows}
\usetikzlibrary{matrix}
\usetikzlibrary{patterns}
\usetikzlibrary{shapes}
\usepackage{authblk}

\usepackage[arc,all]{xy}

\usepackage[left=3cm,top=3cm,right=3cm,bottom=2cm]{geometry}

\def\fCenter{{\mbox{$\ \vdash\ $}}}

\EnableBpAbbreviations

\renewcommand{\epsilon}{\varepsilon}

\tikzset{
	treenode/.style = {align=center, inner sep=0pt, text centered},
	Ske/.style = {treenode, ellipse, double, draw=black,
		minimum width=6pt, thick},
	PIA/.style = {treenode, ellipse, black, draw=black,
		minimum width=6pt},
	Crit/.style = {treenode, rectangle, draw=black,
		minimum width=0.5em, minimum height=0.5em}
}

\theoremstyle{plain}
\newtheorem{thm}{Theorem}

\newtheorem{lem}[thm]{Lemma}

\newtheorem{lemma}[thm]{Lemma}

\theoremstyle{definition}

\DeclareSymbolFont{mysymbol}{U}{txsyc}{m}{n}
\SetSymbolFont{mysymbol}{bold}{U}{txsyc}{bx}{n}
\DeclareFontSubstitution{U}{txsyc}{m}{n}
\DeclareMathSymbol{\bdiam}{\mathord}{mysymbol}{95}

\usepackage{graphicx}

\title{A proof-theoretic approach to abstract interpretation}
\date{}
\author[1]{Vijay D'Silva}
\author[2,3]{Alessandra Palmigiano}
\author[4]{Apostolos Tzimoulis}
\author[5]{Caterina Urban}
\affil[1]{Google Inc., San Francisco}
\affil[2]{Vrije Universiteit Amsterdam, The Netherlands}
\affil[3]{University of Johannesburg, South Africa}
\affil[4]{University of Luxembourg, Luxembourg}
\affil[5]{INRIA, France}

\begin{document}

\maketitle

\section{Introduction}

Abstract interpretation is a theory of formal program verification which generates sound approximations of the semantics of programs, and 
 has been used as the basis of methods and effective algorithms to approximate undecidable or computationally intractable problems such as the verification of safety-critical software (e.g.~medical, nuclear, aviation software).

Typically,  a complex concrete model (such as the powerset $\mathcal{P}(\Sigma)$ of a possibly infinite set modelling program executions) is related to a model that can be efficiently represented and manipulated, which is usually a finite lattice $A$ that encodes the relevant -- logically interconnected -- properties about these executions, by means of an adjoint pair of maps. Specifically, the right adjoint (the {\em concretization} map $\gamma: A\to \mathcal{P}(\Sigma)$)  provides the intended interpretation of the symbolic properties (that is, $S\models a$ iff $S\subseteq \gamma(a)$ for any $S\in \mathcal{P}(\Sigma)$ and  $a\in A$); the left adjoint (the {\em abstraction} map $\alpha: \mathcal{P}(\Sigma)\to A$) classifies the executions of the given program according to their satisfying the relevant properties. Even though the concrete model is usually a Boolean algebra, the algebraic structure that the abstract lattice $A$ retains depends on the properties preserved by the concretization map $\gamma$.    


Although this theory was connected to logic since its inception \cite{cousot1,cousot2,jensen}, it is only in the last decade that the connection was made systematic. In particular, the notion of an (internal) logic of an abstraction was introduced in  \cite{schmidt2008internal}  and systematically related to the order-theoretic properties of the concretization map. In \cite{DSilva2017}, this line of research is further developed. Namely, the logics underlying  specific abstractions are identified, together with explicit specification of  proof-theoretic presentations for each of them.

The present note collects the preliminary results and observations of an ongoing work in which  we generalise the results of \cite{DSilva2017} and introduce a general procedure for generating the (internal) logic of an abstraction together with the specification of a proof system for it. The main idea  is to  generate a logic whose Lindenbaum-Tarski algebra is isomorphic to the  abstract algebra $A$.  

\section{Generating a logic corresponding to a finite abstraction}

We describe a general procedure for generating a logic $\mathcal{L}_A$ with one variable corresponding to a finite abstraction $\mathcal{A}=(A,\sqsubseteq, Op_A)$ with concretization $\gamma:\mathcal{A}\to (P(\text{Struct}),\subseteq, Op_c)$.  

\begin{enumerate}
	\item The logical connectives of the language of $\mathcal{L}_A$ are the connectives preserved by $\gamma$. Notice that, whenever the finite abstraction $\mathcal{A}$ is a distributive lattice,  being  finite, $\mathcal{A}$ will be a bi-Heyting algebra. So in this case, the implication and the co-implication of the abstract lattice  become available candidates as connectives of the language.
	\item for every point $a\in A$,  add a unary predicate symbol $a(x)$ to the language; 
	\item for every connective that is preserved by $\gamma$, add  the introduction rules appropriate to that connective in the proof system;
	\item for every binary connective  $\star$ in $\mathcal{L}_A$ and predicates $a(x),b(x),c(x)$ such that $a\star b=c$,  add a rule corresponding to the axiom $a(x)\star b(x)\dashv\vdash c(x)$ in the proof system;
	\item for every unary connective $\star$ and predicates $a(x)$ and $b(x)$ such that $\star a=b$,  add a rule corresponding to the axiom $\star a(x)\dashv\vdash b(x)$.
	\item for all predicates $a(x)$ and $b(x)$ such that $a\leq b$,  add a rule corresponding to the axiom $a(x)\vdash b(x)$.
\end{enumerate}

\begin{lem}
	The logic $\mathcal{L}_A$ is sound w.r.t.\ the concretization.
\end{lem}
\begin{proof}
	The proof is by induction on the depth of the proof. The axioms are sound and the preservation of the connectives of $\gamma$ guarantee that the introduction of the connectives are sound. Finally by the preservation of the connectives by $\gamma$ and the definition of satisfaction of sequents in the concrete algebra (simply by verifying if the interpretation of the left hand side is included in the interpretation of the right hand side), the rules corresponding to the relationship between predicates described in items 4 and 5 are also sound.
\end{proof}

Let $\mathbb{L}$ be the Lindenbaum-Tarski algebra of $\mathcal{L}_A$.

\begin{lemma}
	The algebra $\mathbb{L}$ is isomorphic to $\mathcal{A}$.
\end{lemma}
\begin{proof}
	Let $e:\mathcal{A}\to\mathbb{L}$ be defined as $e(a)=[a(x)]$, where $[a(x)]$ denotes the equivalence class of the predicate $a(x)$. First notice that $e$ is surjective since each element in the Lindenbaum-Tarski algebra of the logic $\mathcal{L}_A$ is  the equivalence class of some predicate. This follows immediately by induction on the complexity of the formulas and items 4 and 5 in the construction.
	
	To show that $e$ is injective let $[a(x)]=[b(x)]$. Then we have that $a(x)\dashv\vdash b(x)$. Notice that $\mathcal{L}_A$ is sound w.r.t.\ $\mathcal{A}$ when each predicate $c(x)$ is interpreted as $c$. Hence, we immediately get that $a=b$.
	
	Finally, $e$ being an order embedding and a homomorphism w.r.t.\ the connectives of the language follows immediately from items 4, 5 and 6, and the soundness of the logic w.r.t.\ $\mathcal{A}$.
\end{proof}

\begin{lem}
	If $\gamma$ is an order-embedding, then $\mathcal{L}_A$ is complete w.r.t.\ the concretization. 
\end{lem}
\begin{proof}
	Let us assume that $\gamma(a)\leq\gamma(b)$ and show that $a(x)\vdash b(x)$ is derivable. Since $\gamma$ is an order embedding, $\gamma(a)\leq\gamma(b)$ implies  that $a\leq b$. Then by item 6, $a(x)\vdash b(x)$ is a rule of the proof system, and thus it is derivable.
\end{proof}

Notice that the lemma above strengthens  \cite[Proposition 5]{schmidt2008internal}. Indeed, the statement of that proposition  assumes  $\gamma$ to be an injective  right adjoint, which implies but is not implied by the condition that $\gamma$ is an order embedding.

The above procedure is very naive: it is  inefficient, since many axioms will be redundant, and it ignores computational limitations. However, it is applicable to any finite lattice $\mathcal{A}$ and any concretization $\gamma$.

\section{Cartesian products}
The proof theoretic analysis  discussed in the previous section naturally generalizes to $n$ variables. Indeed,  formulas $\varphi(\overline{x})$ for $\overline{x}=(x_1,\ldots, x_n)$ can be interpreted on $\wp(\mathbb{Z}^n)$ as $\{\overline{a}\in\mathbb{Z}^n\mid \varphi(\overline{a})\}$, and  $m$ such predicates $\varphi_1(\overline{x}),\ldots,\varphi_m(\overline{x})$  can be arranged into a finite lattice $\mathbb{A}$, such that the assignment $\varphi(\overline{x})\mapsto \{\overline{a}\in\mathbb{Z}^n\mid \varphi(\overline{a})\}$ is an order embedding and provides the concretization map $\gamma:\mathbb{A}\to\wp(\mathbb{Z}^n)$.

When $\varphi_i(\overline{x})=\psi^i_1(x_1)\land\cdots\land\psi^i_n(x_n)$ and $\varphi_i(\overline{x})\vdash\varphi_j(\overline{x})$ if and only if $\psi^i_k(x_k)\vdash\psi^j_k(x_k)$ for every $1\leq k\leq n$, then $\mathbb{A}=\prod_{1\leq k\leq n}\mathbb{A}_k$, $\gamma=\gamma'\circ\iota$ where $\iota:\wp(\mathbb{Z})^n\to \wp(\mathbb{Z}^n)$ is defined by the assignment $(X_1,\ldots,X_n)\mapsto\prod_{1\leq k\leq n}X_k$ and $\gamma'=\prod_{1\leq k\leq n}\gamma_k$ with $\gamma_k:\mathbb{A}_k\to\wp(\mathbb{Z})$ and $\psi^k_j(x_k)\mapsto\{a\in\mathbb{Z}\mid \psi^k_j(a)\}$.

Intuitively, this means that $\gamma(\varphi(\overline{x}))$ is an $n$-dimensional rectangle in $\wp(\mathbb{Z}^n)$. Notice that $\iota$ has a left adjoint which assigns $R\in\wp( \mathbb{Z}^n)$ to the smallest $n$-dimensional rectangle containing $R$. Hence, $\iota$ always preserves all existing meets. Moreover, if $(X_1,\ldots,X_n)$ is such that $X_i=\varnothing$ for some $1\leq i\leq n$, then $\iota(X_1,\ldots,X_n)=\varnothing$.

\section{On the logic of octagons}

We provide some observations on the logic of predicates of the form $\pm x\pm y\geq c$ where $c\in\mathbb{Z}$ is a constant. Let $\mathcal{A}$ be the abstract algebra corresponding to these predicates and let $\gamma:\mathcal{A}\to\wp(\mathbb{Z}\times\mathbb{Z})$ be the concretization map. 

First of all we notice that  $\gamma$ preserves negation since $$\mathbb{Z}\times\mathbb{Z}\setminus\gamma(\pm x\pm y\geq c)=\gamma(\mp x\mp y\geq -c+1).$$ On the other hand, it cannot preserve conjunction since there does not exist any  of the above predicates that is satisfied by the same $x,y$ such that $x+y\geq 0$ and $x-y\geq 0$. From this, it follows that $\gamma$ cannot preserve implication or disjuntcion (since any of these connectives along with negation are able to define the rest). 

However, the logic of (classical) negation, in the absence of conjunction or disjunction, is very weak. Indeed, the only things we can say about negation is that it is an involution and is order reversing.

Therefore, if we define the logic in the way  described above, the following lattice:

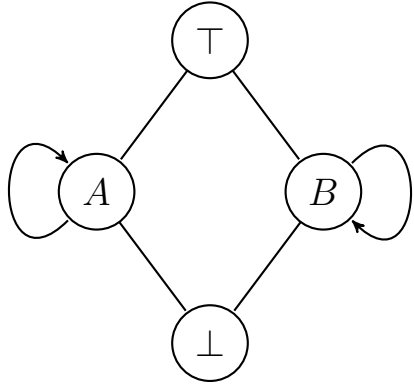
\begin{center}
	
		\centering
		\begin{tikzpicture}[-,>=stealth',shorten >=1pt,thick]
		\SetGraphUnit{3} 
		\tikzset{VertexStyle/.style = {draw,circle,thick,
				minimum size=1cm,
				font=\Large\bfseries},thick}
		\Vertex[L={$\top$},x=0,y=0]{A};
		\Vertex[L={ $A$},x=-1.5,y=-2]{B};
		\Vertex[L={ $B$},x=1.5,y=-2]{C};
		\Vertex[L={ $\bot$},x=0,y=-4]{D};
		\Edge(A)(B);
		\Edge(A)(C);
		\Edge(C)(D);
		\Edge(B)(D);
		
		\Loop[dist=1.5cm,dir=WE](B)
		\Loop[dist=1.5cm,dir=EA](C)
		\end{tikzpicture}
\end{center}

with negation $\lnot(\top)=\bot$, $\lnot(\bot)=\top$, $\lnot(A)=A$ and $\lnot(B)=B$ satisfies the aforementioned properties and hence would be a model of the logic.

To avoid this, we propose to still use conjunction and disjunction as meta-connectives in the logic. In this way we can still consider sequents of the type $$x+y\geq 1, -x-y\geq 0\vdash \mathtt{ff}.$$

Another observation concerning the lattice $\mathcal{A}$ is that every element that is (not top and) not bottom  is meet-irreducible. Indeed, if $p(x,y)$ and $q(x,y)$ are two incomparable such predicates, either they have empty intersection or their intersection is a portion of the plane described by an angle and therefore cannot be described by a single line. Reasoning in a similar way one can show that every element in $\mathcal{A}$ that is not top (and not bottom) is also join-irreducible. Hence $\mathcal{A}$ must have the following shape:

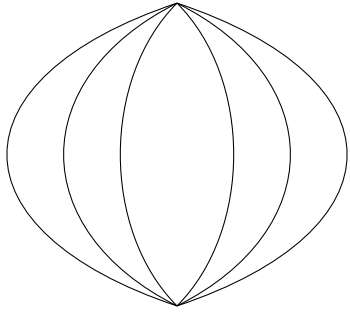
\begin{center}
\begin{tikzpicture}
\draw (0,0) .. controls (-3,1) and (-3,3) .. (0,4);
\draw (0,0) .. controls (3,1) and (3,3) .. (0,4);
\draw (0,0) .. controls (-2,1) and (-2,3) .. (0,4);
\draw (0,0) .. controls (2,1) and (2,3) .. (0,4);
\draw (0,0) .. controls (-1,1) and (-1,3) .. (0,4);
\draw (0,0) .. controls (1,1) and (1,3) .. (0,4);
\end{tikzpicture}
	
\end{center}

Notice that, in such an algebra, an involutive negation similar to the one described above always exists: Just send each point in the north hemisphere to its corresponding point in the south hemisphere and the points at the equator to themselves. 

 The elements discussed so far  enable us to  generate the Lindenbaum-Tarski algebra of this logic in such a way that it is order-isomorphic to $\mathcal{A}$.


\section{Cartesian and non-Cartesian abstractions}

In the previous section, we saw that non-Cartesian abstractions are better behaved when it comes to putting  together the logic of the two variables and their relationships. This suggests a possible strategy, namely:
\begin{itemize}
	\item develop the general theory for the non-Cartesian abstractions;
	\item characterise what percolates to the Cartesian case.
\end{itemize}
We observe that the natural map $$\wp(\mathbb{Z})\times\wp(\mathbb{Z})\rightarrow\wp(\mathbb{Z}\times\mathbb{Z})$$ mapping $(X,Y)$ to $X\times Y$ preserves all meets hence it is a right-adjoint and its left adjoint is the map assigning $R\subseteq\mathbb{Z}\times\mathbb{Z}$ to $\pi_1(R)\times\pi_2(R)$. The question is: Can we use somehow this adjunction situation to extract the logic of two variables as a refinement of the relational logic?

Another observation is that this map is injective when restricted to $$(\wp(\mathbb{Z})\setminus\{\varnothing\})\times(\wp(\mathbb{Z})\setminus\{\varnothing\})$$

\bibliographystyle{abbrv}
 \bibliography{reference}

\end{document}